\documentclass[aps,prl,reprint,twocolumn,superscriptaddress]{revtex4-2}
\pdfoutput=1 %
\usepackage{physics}
\usepackage{amsfonts}
\usepackage{graphicx}
\usepackage[a4paper, margin=1in]{geometry}
\usepackage{amsmath, amsthm}
\usepackage{makecell}
\usepackage{hyperref}
\usepackage[capitalize]{cleveref}
\usepackage{amssymb}
\usepackage{dsfont}
\usepackage{mathtools}
\usepackage{color,xcolor}
\usepackage{multirow}
\usepackage{ragged2e}
\usepackage{subcaption}
\usepackage{orcidlink}
\usepackage{comment}
\newlength{\ketketwidth}
\newlength{\ketwidth}

\newcommand{\kettstylesep}[3]{
    \settowidth{\ketwidth}{$#2\left|#1\right\rangle$}
    \settowidth{\ketketwidth}{$#2\left.\left|#1\right\rangle\right\rangle$}
    \left|#1\right\rangle#3\hspace{\ketwidth}\hspace{-\ketketwidth}
}
\newcommand{\kett}[1]{
    \left.\mathchoice
        {\kettstylesep{#1}{\displaystyle}{\hspace{0.3em}}}
        {\kettstylesep{#1}{\textstyle}{\hspace{0.3em}}}
        {\kettstylesep{#1}{\scriptstyle}{\hspace{0.3em}}}
        {\kettstylesep{#1}{\scriptscriptstyle}{\hspace{0.25em}}}
    \right\rangle
}

\newcommand{\bbrastylesep}[3]{
    \settowidth{\ketwidth}{$#2\left\langle#1\right|$}
    \settowidth{\ketketwidth}{$#2\left\langle\left\langle#1\right|\right.$}
    #3\hspace{\ketwidth}\hspace{-\ketketwidth}\left\langle#1\right|
}
\newcommand{\bbra}[1]{
    \left\langle\mathchoice
        {\bbrastylesep{#1}{\displaystyle}{\hspace{0.3em}}}
        {\bbrastylesep{#1}{\textstyle}{\hspace{0.3em}}}
        {\bbrastylesep{#1}{\scriptstyle}{\hspace{0.3em}}}
        {\bbrastylesep{#1}{\scriptscriptstyle}{\hspace{0.25em}}}
    \right.
}

\newcommand{\kettbbra}[1]{\kett{#1}\hspace{-0.3em}\bbra{#1}}

\newcommand{\one}{\mathds{1}}
\newcommand{\Hil}{\mathcal{H}}
\newcommand\varul[2][3]{\mkern#1mu\underline{\mkern-#1mu#2\mkern-#1mu}\mkern#1mu}
\newcommand{\str}[1]{\varul[0]{#1}}

\definecolor{applegreen}{rgb}{0.55, 0.7, 0.0}

\newcommand{\Suppl}{Supplemental Material}
\newtheorem{definition}{Definition}
\newtheorem{theorem}{Theorem}
\newtheorem{corollary}{Corollary}

\begin{document}

\title{Proper and improper mixed states
serve as different prior beliefs for quantum state retrodiction}

\author{Mingxuan Liu\orcidlink{0009-0009-4361-2641}}
\affiliation{Centre for Quantum Technologies, National University of Singapore, Singapore 117543, Singapore}
\author{Valerio Scarani\orcidlink{0000-0001-5594-5616}}
\affiliation{Centre for Quantum Technologies, National University of Singapore, Singapore 117543, Singapore}
\affiliation{Department of Physics, National University of Singapore, Singapore 117542, Singapore}
\author{Ge Bai\orcidlink{0000-0002-6814-8840}}
\email{gebai@hkust-gz.edu.cn}
\affiliation{Centre for Quantum Technologies, National University of Singapore, Singapore 117543, Singapore}
\affiliation{Thrust of Artificial Intelligence, Information Hub, The Hong Kong University of Science and Technology (Guangzhou), Guangzhou 511453, China}

\date{\today}

\begin{abstract}

A mixed quantum state can be taken as capturing an unspecified form of ignorance; or as describing the lack of knowledge about the true pure state of the system (``proper mixture''); or as arising from entanglement with another system that has been disregarded (``improper mixture''). These different views yield identical density matrices and therefore identical predictions for future measurements. But when used as prior beliefs for inferring the past state from later observations (``retrodiction''), they lead to different updated beliefs. This is a purely quantum feature of Bayesian agency. Based on this observation, we establish a framework for retrodicting on any quantum belief and we prove a necessary and sufficient condition for the equivalence of beliefs. We also illustrate how these differences have operational consequences in quantum state recovery.

\end{abstract}

\maketitle

{\em Introduction---}Mixed states describe incomplete knowledge of the system $S$ under study. The missing pieces of information are also physical, and can be described as a system $R$, correlated with $S$ but to which the agent does not have access. In classical physics, these correlations can always be interpreted as introducing a label: the system $S$ is in one among several possible pure states, but we ignore which one. In quantum physics, there is another possibility: the systems $S$ and $R$ may be entangled, and the joint system may even be in a pure state. D'Espagnat named these two possibilities proper mixtures and improper mixtures, respectively \cite{d1976conceptual}.

Despite their distinct interpretations, proper and improper mixtures are usually thought to be indistinguishable \cite{hughston1993complete,kirkpatrick2001indistinguishability,castellani2023all}: both are represented by identical density matrices and therefore yield the same {\em predictions} for quantum measurements. This indistinguishability has led some to view the distinction as purely interpretative (``going to the Church of the Larger Hilbert Spaces'', as allegedly first quipped by John Smolin), with no operational consequences. However, things become subtle when considering {\em retrodiction}: inference about the past based on current knowledge, here specifically past state inference \cite{barnett-pegg-jeffers,fuchs2002quantum,PhysRevLett.115.180407}, and recovery of irreversible processes \cite{barnum-knill,SutterRennerFawzi_2016,sutter-tomamichel-harrow,PhysRevLett.128.020403,junge,zheng2024near,biswas2024noise}. The common tool for retrodiction is Bayes' rule, for which several quantum extensions have been proposed \cite{caves1986quantum,gardiner2004quantum,luders2006concerning,ozawa1997quantum,fuchs2001quantum,schack2001quantum,warmuth2005bayes,Leifer-Spekkens,tsang2022generalized,parzygnat2022non,surace2023state,parzygnat2023time,parzygnat2023axioms}. Essentially, Bayes' rule is a belief update, and it requires the agent to choose a prior belief about the initial state of the system.

In this Letter, we show that proper and improper mixtures serve as different prior beliefs for quantum state retrodiction. Such a difference is a purely quantum feature, due to the aforementioned possibility of having mixed knowledge about a subsystem while having maximal knowledge on the global system; and to the possibility of different preparations for the same mixture. We further extend this observation beyond the notion of proper and improper mixtures, formulating a general classification theorem of quantum beliefs and including all above mixtures as special cases. Our observation not only refreshes the conceptual understanding of quantum mixtures, but also demonstrates their operational significance in practice.

{\em Distinction between proper and improper mixtures---}We start with an example of one qubit. An agent describes the state prepared by a device as the maximally mixed state $\one/2$ (we don't need to specify whether this description is an \emph{a priori} belief, or arises through some form of tomography). This knowledge is compatible with several scenarios. Let us look at just two:
\begin{enumerate}
    \item\label{case:ensemble} The device is designed to prepare only the states $\ket0$ and $\ket1$. The agent thus believes that, in each round, the state of the system is either $\ket0$ or $\ket1$ with $1/2$ probability each, viewing the mixture as a \textit{proper mixture}. 
    \item\label{case:entangled} The device is an entanglement source that always outputs two qubits $S$ and $R$, but the agent sees only qubit $S$. The agent believes that system to be half of a maximally entangled state $(\ket{00}_{SR}+\ket{11}_{SR})/\sqrt2$, viewing the mixed state of $S$ as an \textit{improper mixture}.
\end{enumerate}

While these two choices of the agent do not affect any prediction, we turn our attention to retrodiction. Suppose that the agent makes a measurement on the prepared state in the Pauli-Z eigenbasis $\{\ket{0},\ket{1}\}$, and wants to infer what has been prepared by the device. In case \ref{case:ensemble}, if the outcome is $+1$ (or $-1$), the agent infers with certainty that the prepared state was $\ket0$ (or $\ket1$). By contrast, in case \ref{case:entangled}, no measurement outcome can modify the belief that the device always prepares a maximally entangled state, because this state is pure. If asked in what state the system $S$ was, the agent will answer $\one/2$ in every round, irrespective of the measurement outcome.

Thus the beliefs in a proper or improper mixture, though identical for all predictions, lead to different updates of our knowledge about the past. That this must be the case follows from the conjunction of two well-known statements. On the one hand, from quantum theory, the purification principle \cite{chiribella2010probabilistic,chiribella2011informational,d2017quantum} states that every mixed state of system $S$ can be seen as the marginal of a pure joint state of $S$ with another system $R$. On the other hand, a tenet of Bayesian inference \cite{pearl,jeffrey, jaynes_2003} is that certainty is immune to updates: no amount of evidence can alter a belief of maximal information (classically, an event with prior probability one). From these, our main observation follows: if a mixture is of unknown origin or is believed to be a proper mixture, the agent admits an element of ignorance that could be updated by further evidence; if the agent's belief is a pure state (albeit involving a system they may know nothing about), this is maximal information and cannot be modified by any further evidence. 

This observation holds for any formulation of the quantum Bayes' rule, since it only requires the property that certainty is immune to updates. In the following, we focus on a specific formulation of quantum Bayes' rule, the Petz map (also known as Petz transpose map and Petz recovery map) \cite{petz1,petz}. We consider the Petz map as the most representative formulation since it is the only known proposal that satisfies axioms derived from classical Bayes' rule \cite{parzygnat2023axioms}. Furthermore, it stems from a principle of minimum belief change when updating with new information \cite{Baimcp}, mirroring how classical Bayes' rule itself arises.

{\em Formulation of quantum beliefs---}We denote by $S(\Hil)$ the set of density matrices on Hilbert space $\Hil$. All quantum processes can be described by completely positive trace-preserving (CPTP) maps.

As a quantum analog of Bayes' rule, the Petz map is a rule for updating quantum beliefs, which are described by density operators. For a process $\mathcal{E}$ from system $S$ to system $T$, the agent chooses an initial belief $\beta_S \in S(\Hil_S)$.
After obtaining the output state $\sigma \in S(\Hil_T)$, the Petz map produces the updated belief about $S$ \cite{petz1,petz}
\begin{align} \label{eq:Petz}
    \mathcal{R}^{\mathcal{E},\beta_S}(\sigma) := \sqrt{\beta_S}\mathcal{E}^\dagger\left(\mathcal{E}(\beta_S)^{-\frac{1}{2}}\sigma\mathcal{E}(\beta_S)^{-\frac{1}{2}}\right) \sqrt{\beta_S}
\end{align}
where $\mathcal{E}^\dag$ is the adjoint map of $\mathcal{E}$ defined by the unique linear map satisfying $\Tr[\mathcal{E}(X)^\dag ~ Y] = \Tr[X^\dag~ \mathcal{E}^\dag(Y)]$, for all operators $X$ and $Y$.

This is the conventional definition of belief update when the quantum Bayes' rule is given by the Petz map. Since it depends on the prior belief on the system alone, it is clearly unable to distinguish the two cases discussed above, as the density matrix of the system is the same in both cases. For that, we need to incorporate possible correlations, classical or quantum, to the description of the belief. Thus, we describe our prior belief as a joint state $\beta$ on $S$ and an additional system $R$. We still work under the assumption that only the output of the process $\mathcal{E}$ acting on $S$ is available to the agent, while system $R$ is hidden. Thus, the complete process is described as $\Tr_R \circ (\mathcal{E}\otimes\textrm{id}) \equiv \mathcal{E} \otimes \Tr$. The resulting belief update on the joint system is then described by the Petz map of the complete process $\mathcal{E} \otimes \Tr$ with the extended prior $\beta$, resulting in the {\em prior-extended Petz map} $\mathcal{R}^{\mathcal{E} \otimes \Tr, \beta}$. The updated belief induced on the system $S$ alone is further obtained by ignoring system $R$ as:
\begin{align}\label{eq:prior-ext_Petz}
    &\mathcal{R}_{\rm ext}^{\mathcal{E},\beta}(\sigma) := (\Tr_R \circ \mathcal{R}^{\mathcal{E} \otimes \Tr, \beta})(\sigma)\\=&\Tr_R\left[\sqrt{\beta}\left(\mathcal{E}^\dag(\mathcal{E}(\beta_S)^{-\frac{1}{2}}\sigma\mathcal{E}(\beta_S)^{-\frac{1}{2}}) \otimes\one_R\right)\sqrt{\beta}\right] \nonumber
\end{align}
where $\beta_S := \Tr_R[\beta]$. For detailed derivation in Kraus form, please refer to \Suppl~\cite{SM}. We will adopt this definition as the \textit{retrodiction map} throughout the remainder of the paper. It generalizes \cref{eq:Petz}, allowing the prior to include correlations with a reference system $R$, and reduces to the original case when no such correlations are present, namely for \(\beta=\beta_S\otimes\beta_R\).

Let us now first revisit our earlier examples of proper and improper mixtures using these tools. The measurement in the $\{\ket0,\ket1\}$ basis can be written as a CPTP map from system $S$ to a classical system as
\begin{align}
    \mathcal{E}_{0/1}(\rho):=\ev{\rho}{0} \ketbra{\str0} + \ev{\rho}{1}\ketbra{\str1}\,,
\end{align}
where $\{\ket{\str0},\ket{\str1}\}$ is a basis of the classical system. A distinction arises when choosing the belief for retrodicting $\mathcal{E}_{0/1}$ as a proper or improper mixture. For the case \ref{case:ensemble} of a proper mixture, since the device prepares either $\ket0$ or $\ket1$, the randomness in the device can be modeled as a perfect coin: the device flips a coin and prepares either $\ket0$ and $\ket1$ according to the outcome. Therefore, we choose the belief as
\begin{align} \label{eq:beta_proper}
    \beta_1 := \frac12\ketbra{0}_S \otimes \ketbra{\str0}_{R_1} + \frac12\ketbra{1}_S \otimes \ketbra{\str1}_{R_1}
\end{align}
where $R_1$ is a classical system that represents the internal coin of the device, which is hidden from the agent.
We use the notation $\ket{\str0}_{R_1},\ket{\str1}_{R_1}$ to emphasize the additional system is classical. Our retrodiction map then yields the belief update
\begin{align}
    \mathcal{R}_{\rm ext}^{\mathcal{E}_{0/1},\beta_1}(\ketbra{\str0}) = \ketbra0, ~\mathcal{R}_{\rm ext}^{\mathcal{E}_{0/1},\beta_1}(\ketbra{\str1}) = \ketbra1,
\end{align}
which matches our earlier discussion. On the other hand, in the case \ref{case:entangled} of the improper mixture, the prior belief is described by the entangled state
\begin{align} \label{eq:beta_improper}
    \beta_2 := \ketbra{\Phi^+}_{SR_2}, ~ \ket{\Phi^+}:=\frac{\ket{00}+\ket{11}}{\sqrt2}
\end{align}
The prior-extended Petz map returns the same state, since it is a pure state; consequently, our retrodiction map maps all outcomes back to the initial belief:
\begin{align}
    \mathcal{R}_{\rm ext}^{\mathcal{E}_{0/1},\beta_2}(\ketbra{\str0})=\mathcal{R}_{\rm ext}^{\mathcal{E}_{0/1},\beta_2}(\ketbra{\str1}) = \frac{\one}{2}
\end{align} We stress again that the fact that a pure state cannot be updated is common to classical and quantum information. What is purely quantum is the fact that this induces no update also for a mixed state, the state of the subsystem under study.

Finally, having an explicit form of the quantum Bayes' rule, we can also discuss another choice of prior belief: the choice of not caring how the ignorance came about and taking the density matrix at face value---in practice here, using $\beta_S:=\one_S/2$ as prior into the original Petz map. This belief behaves as $\beta_1$ for the $\{\ket0,\ket1\}$ measurement, but would be inequivalent if other measurements were considered, for instance that on the Pauli-X eigenbasis $\{\ket+,\ket-\}$. The examples are summarized in \cref{tab:compare3}, together with a fourth one that we'll introduce later.

\begin{table}[]
\renewcommand{\arraystretch}{1.1}
    \centering
    \begin{tabular}{c*{4}{|wc{1cm}}}
    \hline
        \multirow{2}*{Initial belief} & \multicolumn{2}{c|}{Measure 0/1} & \multicolumn{2}{c}{Measure $+/-$} \\
    \cline{2-5}
         & 0 & 1 & $+$ & $-$ \\
    \hline
         No extra system $\beta_S$ & $\ketbra{0}$ & $\ketbra{1}$ & $\ketbra{+}$ & $\ketbra{-}$ \\
    \hline
         $\{\ket0,\ket1\}$ ensemble $\beta_1$ & $\ketbra{0}$ & $\ketbra{1}$ & $\one/2$ &$\one/2$ \\
    \hline
         Improper mixture $\beta_2$ & $\one/2$ &$\one/2$ &$\one/2$ &$\one/2$  \\
    \hline
         Haar random $\beta_{\rm Haar}$ & 
         $\frac{\ketbra{0}+\one}{3}$ & $\frac{\ketbra{1}+\one}{3}$ & $\frac{\ketbra{+}+\one}{3}$ & \rule{0pt}{1.1em}$\frac{\ketbra{-}+\one}{3}$ \\[0.1em]
    \hline
    \end{tabular}
    \caption{\justifying Comparison of updated beliefs on system $S$  after measurements among different initial beliefs. In all cases, the marginal state of the belief on system $S$ is the same: $\beta_S = \Tr_{R_1}[\beta_1]=\Tr_{R_2}[\beta_2] = \Tr_{R}[\beta_{\rm Haar}]:=\one/2$. The most agnostic prior belief without an extra system,  $\beta_S$, is updated to the state revealed by the measurement, for all measurements. By contrast, the proper mixture $\beta_1$ \eqref{eq:beta_proper} is always updated to a state diagonal in the basis $\{\ket0,\ket1\}$. The improper mixture $\beta_2$ \eqref{eq:beta_improper} is not updated because it arises from a pure state, i.e.~from certainty. Update for prior $\beta_{\rm Haar}$ \eqref{eq:Haar_random} shows some bias towards the outcome of the measurement.
    \label{tab:compare3}}
\end{table}

{\em Equivalence between quantum beliefs---}The examples above reveal a structure of quantum beliefs, not restricted to density matrices on the system of interest, but involving an additional latent system that could be classical or quantum. For example, an arbitrary pure state ensemble $\{\ket{\psi_x},p(x)\}$ that averages to the desired density matrix of system $S$ (see Ref.~\cite{hughston1993complete} for a classification of such ensembles), or even an interpolation between proper and improper beliefs.
Other than enumerating various categories of beliefs, we would like a clear identification: what are the possible beliefs that lead to inequivalent retrodictions? 
The following theorem, whose proof is in \Suppl~\cite{SM}, answers this question by giving the equivalence condition between quantum beliefs.

\begin{definition}
    $\beta$ and $\gamma$ are equivalent beliefs if they give the same retrodiction map for all channels, namely $\mathcal{R}_{\rm ext}^{\mathcal{E},\beta}=\mathcal{R}_{\rm ext}^{\mathcal{E},\gamma}$ for all CPTP maps $\mathcal{E}$.
\end{definition}

\begin{theorem}\label{thm:beliefs}
    For channels from $\Hil_S$ to $\Hil_T$, two beliefs $\beta\in S(\Hil_S \otimes \Hil_{R_1})$ and $\gamma\in S(\Hil_S \otimes \Hil_{R_2})$ are equivalent if and only if
    \begin{align} \label{eq:betagamma}
        \Tr_{R_1R_1'}\left[\kettbbra{\sqrt{\beta}}\right] = \Tr_{R_2R_2'}\left[\kettbbra{\sqrt{\gamma}}\right]
    \end{align}
    where the double-ket notation is defined as $\kett{A}:=\sum_{i,j,k,l}\bra{i,j}_{SR}A\ket{k,l}_{SR}\ket{i,j,k,l}_{SRS'R'}$, with systems $S'$ and $R'$ being isomorphic to $S$ and $R$, respectively.
\end{theorem}
Notice that $\beta_S=\gamma_S$ is expected (after all, equivalent beliefs should lead to identical predictions), but we had not imposed it as a constraint in the definition: it follows from the theorem as a necessary condition. Some sufficient conditions for equivalence are:

\begin{enumerate}
    \item\label{case:equ1} $\beta\in S(\Hil_S \otimes \Hil_{R_1})$ and $\beta\otimes\sigma\in S(\Hil_S\otimes \Hil_{R_1} \otimes \Hil_{R_2})$ are equivalent beliefs for any $\sigma\in S(\Hil_{R_2})$. 
    \item\label{case:equ2} For any isometry $V: \Hil_{R_1}\to\Hil_{R_2}$, $\beta\in S(\Hil_S \otimes \Hil_{R_1})$ and $(\one_S \otimes V)\beta(\one_S \otimes V^\dag)\in S(\Hil_S \otimes \Hil_{R_2})$ are equivalent beliefs.
    \item\label{case:equ3} 
    More generally, let $\mathcal{P}$ be a reversible channel acting on system $R$, i.e.\ there exists a CPTP map $\mathcal{Q}$ such that $\mathcal{Q}\circ\mathcal{P}=\textrm{id}_R$, the identity map on system $R$. Then, $\beta$ and $(\textrm{id}_S\otimes\mathcal{P})(\beta)$ are equivalent beliefs.
\end{enumerate}

Condition \ref{case:equ1} states that considering an uncorrelated additional system does not affect retrodiction, as one would expect. Condition \ref{case:equ2} states that a quantum belief is invariant under local isometries on its reference system, indicating that an improper mixture represents a certain quantum belief regardless of which purification is chosen. Condition \ref{case:equ3} is the combination of the former two, because every reversible channel can be decomposed into a tensor product with an additional system followed by an isometry \cite{shirokov2013reversibility}.

A proper mixture is described by giving a state ensemble $\{\ket{\psi_x},p(x)\}$. It corresponds to the belief that the state is in $\ket{\psi_x}$ with confidence $p(x)$. The corresponding quantum belief cannot consist only of the density matrix $\sum_x p(x)\ketbra{\psi_x}$, as it needs to encode classical knowledge, such as the basis of an unobserved measurement or the specification of a probabilistic state preparation device. Rather, it will be described by the joint state $\beta=\sum_x p(x) \ketbra{\psi_x}_S\otimes\ketbra{\str{x}}_R$, where $\{\ket{\str{x}}\}$ is an orthonormal basis of a classical system $R$. For this class of beliefs, the equivalence condition given by \cref{thm:beliefs} leads to the following result.

\begin{corollary} \label{cor:ensemble}
    Two pure state ensembles $\{\ket{\psi_x},p(x)\}$ and $\{\ket{\phi_y},q(y)\}$ are equivalent beliefs if and only if their second moments are equal, namely
    \begin{align} \label{eq:ensemble_equiv}
        \sum_x p(x) \ketbra{\psi_x}^{\otimes2} = \sum_y q(y) \ketbra{\phi_y}^{\otimes2}\,.
    \end{align}
\end{corollary}
See \Suppl{} \cite{SM} for its proof and generalization to ensembles of mixed states. As an example of this corollary, the following ensembles are shown to be equivalent:
\begin{enumerate}
    \item Haar-random distribution over pure qubit states:
    \begin{align} \label{eq:Haar_random}
        \beta_{\rm Haar} := \int_{\rm Haar} d\psi \ketbra{\psi}_S \otimes \ketbra{\str\psi}_R\,,
    \end{align}
    where the integral is over the Haar measure and $\{\ket{\str\psi}\}$ is an orthonormal basis for the infinite-dimensional classical system $R$, satisfying $\braket{\str\psi}{\str{\psi'}}=0$ as long as $\ket{\psi}\neq\ket{\psi'}$.
    \item The pure state ensemble consisting of eigenstates of the three Pauli matrices with uniform probability:
    $\beta_{\rm XYZ} := \frac16\sum_{s\in\mathsf{P}}\ketbra{s}_S\otimes\ketbra{\str{s}}_R$ where $\mathsf{P}:=\{0,1,+,-,+i,-i\}$.
\end{enumerate}

This is because the set $\{\ketbra{s}\}_{s\in\mathsf{P}}$ with uniform probability is a projective 2-design and thus satisfies \cite{klappenecker2005mutually,ambainis2007quantum}
\begin{align}
    \int_{\rm Haar} d\psi \ketbra{\psi} \otimes \ketbra{\psi} =\frac16\sum_{s\in\mathsf{P}}\ketbra{s} \otimes \ketbra{s},
\end{align}
which is \cref{eq:ensemble_equiv} applied to $\rho_{\rm Haar}$ and $\rho_{\rm XYZ}$.
The same conclusion naturally generalizes to other projective 2-designs. The resulting updated beliefs after Pauli-Z or -X measurements are listed in \cref{tab:compare3}.

{\em Consequences---}We have shown the distinction between proper and improper mixture in quantum retrodiction. This distinction between proper and improper mixtures is a unique quantum feature. Classically, all mixtures are proper, and the consideration of a system being alone, or being correlated with another system, makes no difference for Bayesian retrodiction \cite{SM}. Note that quantum retrodiction sometimes also refers to parameter estimation \cite{PhysRevLett.86.2455,PhysRevA.67.032112,GJM13}: fundamentally, this is a combination of the \textit{predictive} formalism of quantum mechanics with classical Bayes' theorem; hence, our findings do not alter any results in this branch.

More generally, different results in retrodiction demonstrated in this Letter rely on different subjective understandings of quantum mixtures.
A pragmatist may take our observation as yet another discussion about foundations that has no practical consequence. 
However, they may lead to differences in \textit{action}, notably when retrodiction is used as a recovery channel after dissipation \cite{PhysRevLett.128.020403,barnum-knill,SutterRennerFawzi_2016,sutter-tomamichel-harrow,junge}.
A notable feature of the Petz map is that it is always CPTP, and thus can be implemented deterministically \cite{wilde-recov,gilyn2020quantum}, with applications to error correction \cite{barnum-knill,zheng2024near,biswas2024noise} and reversing dissipative dynamics \cite{kwon-kim,BS21}. Depending on the application, there is not a unique way to choose the prior belief for the Petz map, and according to our results, a wider range of beliefs involving an additional system may be considered. In \cref{fig:recovery} we show the recovery of a depolarizing channel when the true input state is pure, while the prior belief on the system is $\one/2$ with different extended beliefs. The most agnostic belief is the one that does not consider extra systems. In this case, the unavoidable push towards $\one/2$ (which is both the prior and the fixed point of the map) is small: the recovered state stays rather close to the input state. By contrast, the belief of the Church of the Larger Hilbert spaces (improper mixture) is so dogmatic that the recovered state becomes $\one/2$ irrespective of the input state. The linear combination of these two is explored as a confidence indicator in the \Suppl{}~\cite{SM}. The other two extended beliefs, $\beta_1$ and $\beta_{\rm Haar}$, lead to yet different recoveries shown in \cref{fig:recovery}(b) and \cref{fig:recovery}(d). We leave for future work the investigation of the general effect of extended beliefs and the possible application of these results in various information recovery tasks.

\begin{figure*}

\centering
\includegraphics[]{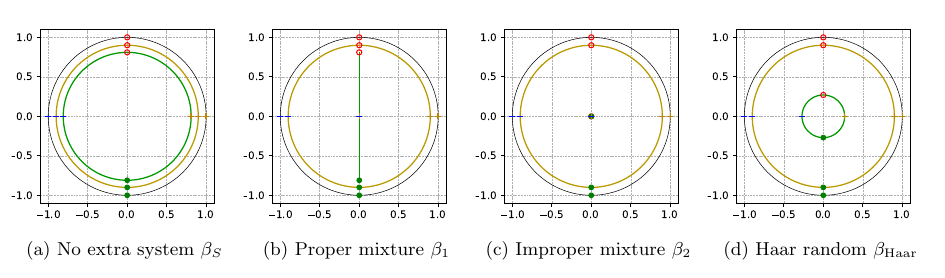}
\caption{\justifying Comparison between different priors for recovering a depolarization error $\mathcal{D}(\rho) := 0.9\rho + 0.1\one/2$. The proper and improper mixtures chosen here are \cref{eq:beta_proper} and \cref{eq:beta_improper} respectively. In all cases, the marginal states of the belief on system $S$ are the same: $\beta_S = \Tr_{R_1}[\beta_1]=\Tr_{R_2}[\beta_2] = \Tr_{R}[\beta_{\rm Haar}] :=\one_S/2$. The figures show the intersection of the Bloch sphere with the x-z plane. The black unit circle represents the original Bloch sphere. The yellow line shows the transformed Bloch sphere after applying $\mathcal{D}$. The green line constitutes the recovered states $(\mathcal{R}_{\rm ext}^{\mathcal{D},\beta}\circ\mathcal{D})(\rho)$. The red circles, green dots, yellow plus signs and blue minus signs represent the states $\ket0,\ket1,\ket+,\ket-$ as prepared, after the channel and after the recovery, respectively. \label{fig:recovery}}
\end{figure*}

{\em Summary and outlook---}We have observed that beliefs about how a mixture arises, indistinguishable from the predictive perspective, serve as fundamentally different priors for retrodiction in quantum information. This is due to purely quantum features: the fact that different preparations may lead to the same mixture, and the fact that a mixture can be even seen as the partial state of a joint pure state. In the latter case, strikingly, no later evidence can update the belief, since it is rooted in a state of maximal knowledge. 
The above observations led us to extend quantum beliefs to include latent correlations. To this effect, we introduced the prior-extended Petz map and the corresponding retrodiction map on the system. We also formulated an equivalence condition for quantum beliefs. 

This study not only enriches the conceptual understanding of quantum mixtures, but also opens new avenues for exploring the role of subjective beliefs and agency in quantum information. Notably, the framework established here has been shown to underlie a unified theory of quantum smoothing~\cite{LiuBaiScarani2025unifying}, which subsumes earlier approaches~\cite{PhysRevLett.115.180407,LiuLaverick2025smoothing}. Popescu proposed a form of fraud detection based on preparation information and called for a formalism to represent such knowledge \cite{popescu2018quantumstatesknowledgepure}. The extended prior offers a natural realization of this idea and may enable quantitative generalization. Our framework will also affect quantum recovery maps beyond the Petz map \cite{parzygnat2023time,surace2023state,parzygnat2023axioms,Baimcp} and have implications for the design and analysis of quantum Bayesian networks~\cite{LEIFER20081899,tucciQBN}, due to its ability to capture quantum knowledge and correlations with additional systems more subtly. Shortly after our paper was made public, we notice that Ref.~\cite{molmerpinske} also discusses different mixtures in the context of retrodiction. However, their aim is to witness system–environment entanglement when allowing joint measurements on both system and environment, which is radically different from ours.

\medskip

\begin{acknowledgments}
This project is supported by the National Research Foundation, Singapore through the National Quantum Office, hosted in A*STAR, under its Centre for Quantum Technologies Funding Initiative (S24Q2d0009). V.S.~acknowledges support from the Ministry of Education, Singapore, under the Tier 2 grant ``Bayesian approach to irreversibility'' (Grant No.~MOE-T2EP50123-0002). G.B.~acknowledges support from the Start-up Fund (Grant No.~G0101000274) from The Hong Kong University of Science and Technology (Guangzhou).
\end{acknowledgments}

\bibliography{refs}
\end{document}


\title{Supplemental Material: Proper and improper mixed states
serve as different prior beliefs for quantum state retrodiction}

\author{Mingxuan Liu}
\affiliation{Centre for Quantum Technologies, National University of Singapore, Singapore 117543, Singapore}
\author{Valerio Scarani}
\affiliation{Centre for Quantum Technologies, National University of Singapore, Singapore 117543, Singapore}
\affiliation{Department of Physics, National University of Singapore, Singapore 117542, Singapore}
\author{Ge Bai}
\email{gebai@hkust-gz.edu.cn}
\affiliation{Centre for Quantum Technologies, National University of Singapore, Singapore 117543, Singapore}
\affiliation{Thrust of Artificial Intelligence, Information Hub, The Hong Kong University of Science and Technology (Guangzhou), Guangzhou 511453, China}
\maketitle

\section{Derivation of retrodiction map with Kraus representation}

In the main text we defined the extended forward map by
\begin{align}
    \mathcal{M} \;:=\; \mathcal{E}\otimes\mathrm{Tr}_R :
\; \mathcal{L}(\mathcal{H}_S\otimes\mathcal{H}_R)\longrightarrow \mathcal{L}(\mathcal{H}_T)
\end{align}
Here we give an explicit Kraus representation of \(\mathcal{M}\) to show how Eq.~(2) in the main text follows from the standard Petz map applied to \(\mathcal{M}\) and \(\beta\).

Let \(\{K_i\}_i\) be a Kraus representation of \(\mathcal{E}\), so that \(\mathcal{E}(\bullet)=\sum_i K_i\,\bullet\,K_i^\dagger\). Choose an orthonormal basis \(\{\ket{j}_R\}_j\) for \(\mathcal{H}_R\). The extended forward map \(\mathcal{M}\) can be expressed as
\begin{equation}
    \begin{split}
        \mathcal{M}(\bullet)
&= \mathrm{Tr}_R\circ(\mathcal{E}\otimes \text{id})\left(\bullet\right)
= \sum_j (\one_S\otimes\bra{j}_R)\sum_i(K_i\otimes \one_R) \bullet (K_i^\dagger\otimes \one_R)(\one_S\otimes\ket{j}_R)\\
&= \sum_{i,j} \big(K_i\otimes\bra{j}_R\big)\,\bullet\,\big(K_i\otimes\bra{j}_R\big)^\dagger
    \end{split}
\end{equation}
Hence a valid set of Kraus operators for \(\mathcal{M}\) is
\begin{equation}\label{eq:Kraus-M}
A_{i,j}\;=\;K_i\otimes\bra{j}_R
\end{equation}
The adjoint map \(\mathcal{M}^\dagger\) is therefore by definition 
\begin{align}
    \mathcal{M}^\dagger(\bullet)=\sum_{i,j} A_{i,j}^\dagger \bullet A_{i,j} = \sum_i (K_i^\dagger \bullet K_i)\otimes \sum_j \ketbra{j}{j}_R=\mathcal{E}^\dagger(\bullet)\otimes \one_R
\end{align}

Now consider the Petz map associated to \(\mathcal{M}\) and to the joint prior \(\beta\) on \(S\otimes R\):
\begin{align}
    \mathcal{R}^{\mathcal{M},\beta}(\bullet)
= \sqrt{\beta}\,\mathcal{M}^\dagger\!\Big(\mathcal{M}(\beta)^{-\tfrac12}\,\bullet\,\mathcal{M}(\beta)^{-\tfrac12}\Big)\sqrt{\beta}
\end{align}
Observe also that \(\mathcal{M}(\beta)=\mathcal{E}(\beta_S)\), where \(\beta_S=\Tr_R[\beta]\). The expression above reduces to
\begin{align}\label{eq:prior-ext_Petz}
    \mathcal{R}^{\mathcal{M},\beta}(\bullet)=\sqrt{\beta}\left(\mathcal{E}^\dag(\mathcal{E}(\beta_S)^{-\frac{1}{2}}\bullet\mathcal{E}(\beta_S)^{-\frac{1}{2}}) \otimes\one_R\right)\sqrt{\beta}
\end{align}
Finally, tracing out the reference system $R$ yields exactly the form given in Eq.~(2) of the main text. It is indeed a CPTP map with the Kraus operators given by
\begin{equation}\label{eq:Kraus-Eq2}
B_{i,j,j'}\;=\big(\one_S \otimes \bra{j'}_R\big)\sqrt{\beta}\big(K_i^\dagger \mathcal{E}(\beta_S)^{-\frac12}\otimes\ket{j}_R\big)
\end{equation}
which satisfy $\sum_{i,j,j'}B_{i,j,j'}^\dagger B_{i,j,j'}=\one_T$ under the standard assumption that $\mathcal{E}(\beta_S)$ has full rank. 
\section{Proof of retrodictive equivalence for classical extended priors }
In classical Bayesian inference, correlation between prior belief of the system $S$ and other system $R$ does not affect retrodiction. This can be observed by assuming all the operators in Eq.~(2) in the main text to commute, or can be understood with classical Bayes' rule. 

Consider a prior belief $\gamma(a)$ defined solely on the system $S$, and a forward process with transition probabilities from $a$ to $b$ being $\varphi(b|a)$. For an observation $b=b_0$, the Bayes' rule gives the probability
\begin{align}
    q_1(a|b_0) = \frac{\varphi(b_0|a)\gamma(a)}{\sum_{a'} \varphi(b_0|a')\gamma(a')}
\end{align}
To match the quantum reverse process, here we use a generalization of the Bayes' rule when the observation on $b_0$ is not given by a single value, but a soft evidence $r(b)$, which is a probability distribution. The posterior probability of system $S$ is given by \cite{jeffrey,pearl,CHAN200567}:
\begin{align}
    q_1(a) = \sum_b\frac{\varphi(b|a)\gamma(a)}{\sum_{a'} \varphi(b|a')\gamma(a')} r(b) 
\end{align}
Now, consider another prior belief $\gamma(a,c)$ that is a joint distribution over the systems $S$ and $R$. Since the value of $b$ depends only on $a$, after incorporating the value $c$ of system $R$, the forward transition probabilities are $\Phi(b|a,c)= \varphi(b|a)$. The posterior distribution is then:
\begin{equation}
\begin{split}
    q_2(a,c) &= \sum_b \frac{\Phi(b|a,c)\gamma(a,c)}{\sum_{a',c'} \Phi(b|a',c')\gamma(a',c')} r(b) \\
    & = \sum_b \frac{\varphi(b|a)\gamma(a,c)}{\sum_{a',c'} \varphi(b|a') \gamma(a',c')} r(b) \\
    & = \sum_b \frac{\varphi(b|a)\gamma(a,c)}{\sum_{a'} \varphi(b|a') \gamma(a')} r(b) 
\end{split}
\end{equation}
By summing over $c$, it is evident that
\begin{align}
    q_2(a) = \sum_c q_2(a,c)=q_1(a)
\end{align}
Thus, any extended prior belief $\gamma(a,c)$, which includes correlations with an auxiliary system $R$, yields the same retrodictive outcome as the marginal prior $\gamma(a)$ defined solely on $S$.

\section{Equivalence condition for quantum beliefs}

\begin{proof}[Proof of Theorem 1]
    We first prove the necessity: if $\beta$ and $\gamma$ are equivalent, then Eq.~(8) in the main text  holds.

    Take an informationally complete POVM $\{F_k\}_{k=1}^{d_S^2}$ in system $S$. Take the following set of channels:
    \begin{align}
        \mathcal{E}_0(\rho) &:= \Tr[\rho]\one_T/d_T\\
        \mathcal{E}_k(\rho) &:= \frac12 \Tr[\rho] \one_T/d_T  + \frac12 \Tr[F_k \rho]\ketbra{0} + \frac12 \Tr[(\one_S - F_k)\rho]\ketbra{1}
    \end{align}
    where $\ket{0},\ket{1}$ are two perfectly distinguishable states in system $T$.
    These channels have the property that $\mathcal{E}_{k}(\rho)$ is full-rank for any $k$ and density operator $\rho$ since $\Tr[F_k \rho]\geq0$ and $\Tr[(\one_S - F_k)\rho]\geq0$.
    
    Since $\beta$ and $\gamma$ are equivalent, they should give the same Petz map for any channel $\mathcal{E}_k$. From Eq.~(2) in the main text , that is to say, for any Hermitian operator $\sigma$,
    \begin{align}\label{eq:same_petz}
        \Tr_{R_1}\left[\sqrt{\beta}\left(\mathcal{E}_k^\dag\left(\mathcal{E}_k(\beta_S)^{-\frac{1}{2}} \sigma \mathcal{E}_k(\beta_S)^{-\frac{1}{2}}\right) \otimes \one_{R_1}\right)\sqrt{\beta}\right] = \Tr_{R_2}\left[\sqrt{\gamma}\left(\mathcal{E}_k^\dag\left(\mathcal{E}_k(\gamma_S)^{-\frac{1}{2}} \sigma \mathcal{E}_k(\gamma_S)^{-\frac{1}{2}}\right) \otimes \one_{R_2}\right)\sqrt{\gamma}\right]
    \end{align}
    Notice that $\mathcal{E}_0^\dag(\tau) = \Tr[\tau] \one_S/d_T$, so for $k=0$ and $\Tr[\sigma]\neq0$, the above equation gives
    \begin{align}
        \Tr[\sigma]\Tr_{R_1}[\beta] = \Tr[\sigma]\Tr_{R_2}[\gamma]
    \end{align}
    which implies $\beta_S=\gamma_S$. For $k=1,\dots,d_S^2$, consider \cref{eq:same_petz} taking $\sigma=\mathcal{E}_k(\beta_S)^{1/2}\ketbra{0}\mathcal{E}_k(\beta_S)^{1/2}$:
    \begin{align}
        \Tr_{R_1}\left[\sqrt{\beta}\left(\mathcal{E}_k^\dag\left(\ketbra{0}\right) \otimes \one_{R_1}\right)\sqrt{\beta}\right] =  \Tr_{R_2}\left[\sqrt{\gamma}\left(\mathcal{E}_k^\dag\left(\ketbra{0}\right) \otimes \one_{R_2}\right)\sqrt{\gamma}\right]
    \end{align}

    Since $\mathcal{E}_k^\dag(\ketbra{0}) = \frac12 \one_S / d_T + \frac12 F_k$, one has
    \begin{align}
        \frac1{2d_T}\Tr_{R_1}[\beta]+\frac12\Tr_{R_1}\left[\sqrt{\beta}\left(F_k \otimes \one_{R_1}\right)\sqrt{\beta}\right] =  \frac1{2d_T}\Tr_{R_2}[\gamma]+\frac12\Tr_{R_2}\left[\sqrt{\gamma}\left(F_k \otimes \one_{R_2}\right)\sqrt{\gamma}\right]
    \end{align}
    and by $\beta_S=\gamma_S$,
    \begin{align}
        \Tr_{R_1}\left[\sqrt{\beta}\left(F_k \otimes \one_{R_1}\right)\sqrt{\beta}\right] =  \Tr_{R_2}\left[\sqrt{\gamma}\left(F_k \otimes \one_{R_2}\right)\sqrt{\gamma}\right]
    \end{align}
    Since $\{F_k\}$ is informationally complete, for any $i,j$, there exist scalars $c_{ijk}$ satisfying $\ketbra{i}{j} = \sum_k c_{ijk} F_k$. Therefore, for any $i,j$,
    \begin{equation}
        \begin{split}
        \Tr_{R_1}\left[\sqrt{\beta}\left(\ketbra{i}{j} \otimes \one_{R_1}\right)\sqrt{\beta}\right] &=\sum_k c_{ijk} \Tr_{R_1}\left[\sqrt{\beta}\left(F_k \otimes \one_{R_1}\right)\sqrt{\beta}\right] \\
        &= \sum_k c_{ijk} \Tr_{R_2}\left[\sqrt{\gamma}\left(F_k \otimes \one_{R_2}\right)\sqrt{\gamma}\right] \\&=  \Tr_{R_2}\left[\sqrt{\gamma}\left(\ketbra{i}{j} \otimes \one_{R_2}\right)\sqrt{\gamma}\right]
        \end{split}
    \end{equation}
    and thus
    \begin{align} \label{eq:betagamma_traceform}
        \sum_{i,j} \Tr_{R_1}\left[\sqrt{\beta}\left(\ketbra{i}{j} \otimes \one_{R_1}\right)\sqrt{\beta}\right]\otimes \ketbra{i}{j}_{S'} =  \sum_{i,j} \Tr_{R_2}\left[\sqrt{\gamma}\left(\ketbra{i}{j} \otimes \one_{R_2}\right)\sqrt{\gamma}\right]\otimes \ketbra{i}{j}_{S'}
    \end{align}
    This equation is equivalent to Eq.~(8) in the main text. To show this, note the left hand side of Eq.~(8) is
    \begin{equation}
        \begin{split}
            \Tr_{R_1R_1'}\left[\kettbbra{\sqrt{\beta}}\right] &= \sum_{i,j,k,l,i',k'} \bra{i,j}\sqrt{\beta}\ket{k,l}\left(\bra{i',j}\sqrt{\beta}\ket{k',l} \right)^*\ketbra{i,k}{i',k'}_{SS'}\\
        & = \sum_{i,j,k,l,i',k'} \bra{i,j}\sqrt{\beta}\ket{k,l}\bra{k',l}\sqrt{\beta}\ket{i',j} \ketbra{i,k}{i',k'}_{SS'}\\
        & = \sum_{i,j,k,i',k'} \bra{i,j}\sqrt{\beta}(\ketbra{k}{k'}\otimes\one_{R_1})\sqrt{\beta}\ket{i',j} \ketbra{i,k}{i',k'}_{SS'} \\
        & = \sum_{i,j,k,i',k'} (\ketbra{i}\otimes\bra{j})\sqrt{\beta}(\ketbra{k}{k'}\otimes\one_{R_1})\sqrt{\beta}(\ketbra{i'}\otimes \ket{j}) \otimes \ketbra{k}{k'}_{SS'} \\
        & = \sum_{k,k'} \Tr_{R_1}\left[\sqrt{\beta}(\ketbra{k}{k'}\otimes\one_{R_1})\sqrt{\beta}\right] \otimes \ketbra{k}{k'}_{SS'}
        \end{split}
    \end{equation}
    which is equal to the left hand side of \cref{eq:betagamma_traceform}. Similarly, the right hand sides of Eq.~(8) and \cref{eq:betagamma_traceform} are equal, too. Therefore, Eq.~(8) and \cref{eq:betagamma_traceform} are equivalent, and we have proved the necessity.

    Next, we show the sufficiency: if Eq.~(8) holds, then $\beta$ and $\gamma$ are equivalent.

    First, taking the partial trace over $S'$ in both sides of Eq.~(8), we obtain $\beta_S=\gamma_S$.

    Second, Eq.~(8) implies \cref{eq:betagamma_traceform}, which further implies that for any operator $F$, one has
    \begin{align} \label{eq:beta_F_gamma}
        \Tr_{R_1}\left[\sqrt{\beta}\left(F \otimes \one_{R_1}\right)\sqrt{\beta}\right] =  \Tr_{R_2}\left[\sqrt{\gamma}\left(F \otimes \one_{R_2}\right)\sqrt{\gamma}\right]
    \end{align}

    Third, for any given channel $\mathcal{E}$ from $S$ to $T$ and any state $\sigma \in S(\Hil_T)$, let $F = \mathcal{E}^\dag\left(\mathcal{E}(\beta_S)^{-1/2} \sigma \mathcal{E}(\beta_S)^{-1/2} \right) = \mathcal{E}^\dag\left(\mathcal{E}(\gamma_S)^{-1/2} \sigma \mathcal{E}(\gamma_S)^{-1/2} \right)$, then \cref{eq:beta_F_gamma} gives
    \begin{align} \label{eq:same_petz_2}
        \Tr_{R_1}\left[\sqrt{\beta}\left(\mathcal{E}^\dag\left(\mathcal{E}(\beta_S)^{-\frac{1}{2}} \sigma \mathcal{E}(\beta_S)^{-\frac{1}{2}}\right) \otimes \one_{R_1}\right)\sqrt{\beta}\right] = \Tr_{R_2}\left[\sqrt{\gamma}\left(\mathcal{E}^\dag\left(\mathcal{E}(\gamma_S)^{-\frac{1}{2}} \sigma \mathcal{E}(\gamma_S)^{-\frac{1}{2}}\right) \otimes \one_{R_2}\right)\sqrt{\gamma}\right]
    \end{align}
    \cref{eq:same_petz_2} holding for any $\mathcal{E}$ and $\sigma$ indicates that $\beta$ and $\gamma$ are equivalent quantum beliefs.
    
\end{proof}

\section{Equivalence condition for ensembles of states}
Here, we consider a generalization to ensembles of possibly mixed states, such as $\{\rho_x,p(x)\}$ with $\rho_x$ having confidence $p(x)$. This may be described by the joint state $\beta=\sum_x p(x) \rho_x \otimes\ketbra{\str{x}}_R$, where $\{\ket{\str{x}}\}$ is an orthonormal basis of classical system $R$. The equivalence condition between such beliefs is the following.
\begin{corollary}\label{cor:mixed_ensemble}
    Two ensembles $\{\rho_x,p(x)\}$ and $\{\sigma_y,q(y)\}$ are equivalent beliefs if and only if
    \begin{align} \label{eq:mixed_ensemble_equiv}
        \sum_x p(x) \sqrt{\rho_x} \otimes\sqrt{\rho_x} = \sum_y q(y) \sqrt{\sigma_y}\otimes\sqrt{\sigma_y}
    \end{align}
\end{corollary}
\begin{proof}
The two beliefs are represented by the joint states
\begin{align} 
    \beta_1 := \sum_x p(x) \rho_x \otimes \ketbra{\str{x}}_{R_1}, \quad
    \beta_2 := \sum_y q(y) \sigma_y \otimes \ketbra{\str{y}}_{R_2}
\end{align}

\begin{equation}
    \begin{split}
        \Tr_{R_1R_1'}\left[\kettbbra{\sqrt{\beta_1}}\right] &= \Tr_{R_1 R_1'}\left[ \sum_{x,x'} \kett{\sqrt{p(x)} \sqrt{\rho_x} \otimes \ketbra{\str{x}}_{R_1}} \bbra{\sqrt{p(x')}\sqrt{\rho_{x'}} \otimes \ketbra{\str{x'}}_{R_1}} \right] \\
    &= \Tr_{R_1 R_1'}\left[ \sum_{x,x'}\sqrt{p(x)p(x')} \kett{ \sqrt{\rho_x}} \bbra{\sqrt{\rho_{x'}} }  \otimes \ket{\str{x}\str{x}}_{R_1R_1'} \bra{\str{x'}\str{x'}}_{R_1R_1'} \right] \\
    &= \sum_{x,x'} \sqrt{p(x)p(x')} \kett{ \sqrt{\rho_x}} \bbra{\sqrt{\rho_{x'}} } \delta_{xx'} \\
    & = \sum_{x} p(x) \kettbbra{\sqrt{\rho_x}} 
    \end{split}
\end{equation}
Similarly,
\begin{align}
    \Tr_{R_2R_2'}\left[\kettbbra{\sqrt{\beta_2}}\right] = \sum_y q(y) \kettbbra{\sqrt{\sigma_y}} 
\end{align}
The equivalence condition in Theorem 1 is therefore equivalent to
\begin{align}
    & \sum_{x} p(x) \kettbbra{\sqrt{\rho_x}} = \sum_y q(y) \kettbbra{\sqrt{\sigma_y}} \\
    \Longleftrightarrow\quad &\sum_{x} p(x) \sqrt{\rho_x}\otimes\sqrt{\rho_x}^T = \sum_y q(y) \sqrt{\sigma_y}\otimes\sqrt{\sigma_y}^T \\
    \Longleftrightarrow\quad & \sum_{x} p(x) \sqrt{\rho_x}\otimes\sqrt{\rho_x} = \sum_y q(y) \sqrt{\sigma_y}\otimes\sqrt{\sigma_y}
\end{align}
which is \cref{eq:mixed_ensemble_equiv}.

\end{proof}

Taking $\rho_x=\ketbra{\psi_x}$ and $\sigma_y = \ketbra{\phi_y}$ in \cref{cor:mixed_ensemble} gives Corollary 1 in the main text.

\section{Example: Confidence characterization of quantum belief}

As clear in the main text, our approach is not limited to the extreme cases of proper and improper mixtures. An agent with a prior belief $\beta_S$ about a quantum system $S$ may further characterize the degree of confidence using our framework via a linear interpolation between the most agnostic belief (without an extra system) and the most certain one (an improper mixture) as 

\begin{equation}\label{eq:confidence}
    \beta(p)=p\ketbra{\Phi}{\Phi}+\frac{1-p}{d_R}\beta_S\otimes\mathds{1}_R
\end{equation}
where $p\in[0,1]$ and $\ketbra{\Phi}{\Phi} \in S(\Hil_S \otimes \Hil_{R})$ is a purification of $\beta_S$. Note that all choices of purification yield equivalent beliefs for this example. Taking $p$ closer to 0 represents a less confident belief that is more sensitive to updates, while taking $p$ closer to 1 represents a more confident belief that is more resistant to updates.

In the context of quantum state recovery, the retrodictive map can act as a recovery channel. When the same prior state $\beta_S$ is held with varying confidence levels $p$ on the preparation, the resulting recovery maps differ. In \cref{fig:confidence} we show how the recovery of a depolarizing channel transforms the Bloch sphere for different value of $p$ in \cref{eq:confidence}. As expected, when $p$ increases, the push towards the prior belief $\one/2$ becomes bigger. For $p=0, 1$, the belief $\beta(p)$ reduces to the most agnostic $\beta_S$ and the improper mixture $\beta_2$, corresponding to panels (a) and (c), respectively, of Fig.~1 of the main text.

\begin{figure}[h]
\includegraphics[scale=0.7]{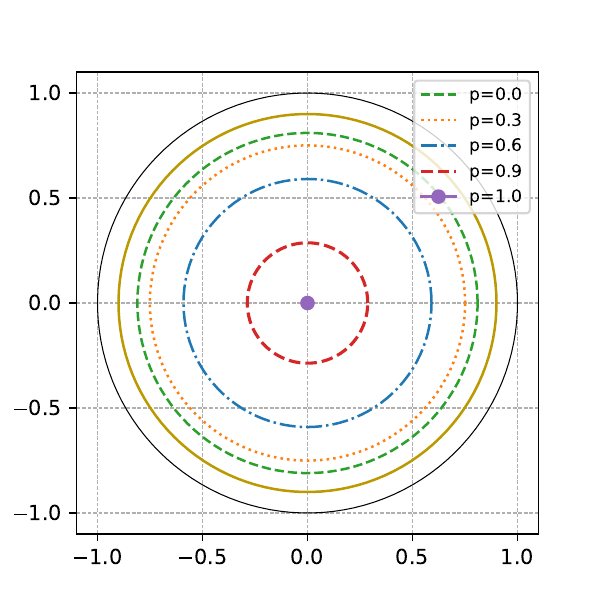}
\caption{\justifying
Comparison between different confidence for recovering a depolarization error $\mathcal{D}(\rho) := 0.9\rho + 0.1\one/2$. In all cases, the marginal states of the belief on system $S$ are the same: $\beta_S=\one_S/2$. The figures show the intersection of the Bloch sphere with the x-z plane. The black line is the original Bloch sphere. The yellow line shows the transformed Bloch sphere after applying $\mathcal{D}$. The dashed lines constitutes the recovered states $(\mathcal{R}_{\rm ext}^{\mathcal{D},\beta(p)}\circ\mathcal{D})(\rho)$.
}
\label{fig:confidence}
\end{figure}

To further analyze the quality of the recovery in the presence of possibly incorrect information, we study an example under two scenarios:
\begin{enumerate}
    \item The forward channel ${\cal F}$ is fully known, but the prior belief $\beta_S$ is incorrect.
    \item The prior belief $\beta_S$ is correct, but the forward channel ${\cal F}$ is only partially known.
\end{enumerate}
We illustrate these cases using a simple model.

Let the initial state be $\rho_0 = 0.9 \ketbra{0}{0} + 0.1 \ketbra{1}{1}$, subjected to a single-qubit depolarizing channel $\mathcal{D}_{\lambda}(\rho) = (1 - \lambda) \rho + \lambda \mathds{1}/2$. Recovery is performed using the retrodictive map $\mathcal{R}_{\rm ext}^{\mathcal{F},\beta(p)}$. The fidelity between the recovered state and the original state is shown in Fig.~\ref{fig:fidelity_recovery} for both scenarios.

Panel (a) illustrates the first scenario where the agent begins with a maximally mixed prior belief, $\beta_S^{(a)} = \frac{\mathds{1}}{2}$. The corresponding extended prior, $\beta^{(a)}(p)$, is constructed according to the definition in \cref{eq:confidence}. In this case, the recovery fidelity, $F(\mathcal{R}_{\mathrm{ext}}^{\mathcal{D}_\lambda, \beta^{(a)}(p)} \circ \mathcal{D}_\lambda(\rho_0), \rho_0)$, declines as the confidence parameter $p$ increases. This reflects the fact that higher confidence in an incorrect prior belief leads to worse recovery performance.

Panel (b) illustrates the second situation where the agent holds the correct prior belief, $\beta_S^{(b)} = \rho_0$, but assumes only partial knowledge of the forward channel by fixing $\mathcal{F} = \mathcal{D}_{0.1}$. The extended prior $\beta^{(b)}(p)$ is again defined via \cref{eq:confidence}. The fidelity in this case is computed as $F(\mathcal{R}_{\mathrm{ext}}^{\mathcal{D}_{0.1}, \beta^{(b)}(p)} \circ \mathcal{D}_\lambda(\rho_0), \rho_0)$. As expected, increased confidence in a correct prior belief enhances the fidelity of recovery.

\begin{figure}[h]
\includegraphics[scale=0.6]{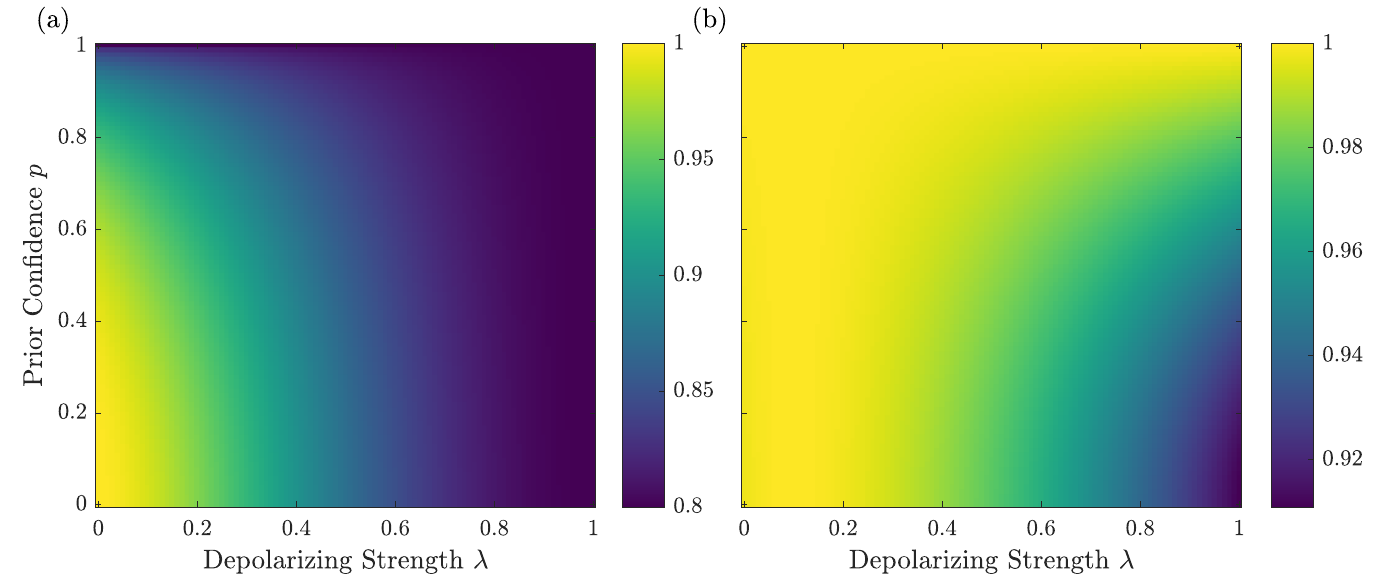}
\caption{\justifying
Recovery fidelity for a single-qubit state $\rho_0$ undergoing a depolarizing channel $\mathcal{D}_{\lambda}$, followed by retrodictive map using $\mathcal{R}_{\rm ext}^{\mathcal{F},\beta(p)}$. (a) Recovery using a maximally mixed prior. (b) Recovery using the correct prior but fixed channel assumption.
}
\label{fig:fidelity_recovery}
\end{figure}

These results, shown in Fig.~\ref{fig:fidelity_recovery}, demonstrate the interplay between prior confidence and recovery performance under our framework.

\bibliography{refs}